\documentclass{article}

\usepackage[reset,a4paper,vmargin=3.5truecm,hmargin=3truecm]{geometry}
\usepackage{amsmath,amsthm,amssymb,amscd,epic}
\usepackage{graphicx,bbm,amsthm,graphicx}
\usepackage{mathtools}
\usepackage{leftidx}
\usepackage{hyperref}
\usepackage{color}
\usepackage{cool}
\usepackage{bm}
\usepackage{bbold}
\theoremstyle{plain}
\newtheorem{thm}{Theorem}[section]
\newtheorem{lem}[thm]{Lemma}
\newtheorem{prop}[thm]{Proposition}
\newtheorem{cor}[thm]{Corollary}

\theoremstyle{definition}
\newtheorem{dfn}{Definition}[section]
\newtheorem{ex}{Example}[section]

\theoremstyle{remark}
\newtheorem{rem}{Remark}[section]

\DeclareMathOperator{\tr}{tr}
\DeclareMathOperator{\Spec}{Spec}

\DeclareMathOperator*{\regprod}{\mathchoice%
{\ooalign{\hbox{$\displaystyle\prod$}\crcr\hbox{$\displaystyle\coprod$}}}
{\ooalign{\hbox{$\textstyle\prod$}\crcr\hbox{$\textstyle\coprod$}}}
{\ooalign{\hbox{$\scriptstyle\prod$}\crcr\hbox{$\scriptstyle\coprod$}}}
{\ooalign{\hbox{$\scriptscriptstyle\prod$}\crcr\hbox{$\scriptscriptstyle\coprod$}}}}

\newcommand{\setC}[3]{ \mathcal{C}^{(#1)}_{#2 #3} }

\newcommand{\bs}{\mathbf{s}}
\newcommand{\br}{\mathbf{r}}
\newcommand{\bZ}{\mathbf{0}}

\newcommand{\HRabi}[1]{H_{\text{\upshape R}}^{#1}} % A. Q. Rabi Model Hamiltonian

\newcommand{\Z}{\mathbb{Z}} % integers
\newcommand{\R}{\mathbb{R}} % real numbers
\newcommand{\C}{\mathbb{C}} % complex numbers

\newcommand{\mat}[1]{\begin{bmatrix}#1\end{bmatrix}}

\newcommand{\matrixZZ}{
  \mat{ \phantom{-}1 & -1 \\ -1 & \phantom{-}1 }
}

\newcommand{\matrixZO}{
  \mat{ -1 & -1 \\ \phantom{-}1 & \phantom{-}1 }
}
\newcommand{\matrixOZ}{
  \mat{ -1 & \phantom{-}1 \\ -1 & \phantom{-}1 }
}
\newcommand{\matrixOO}{
  \mat{ \phantom{-}1 & \phantom{-}1 \\ \phantom{-}1 & \phantom{-}1 }
}

\DeclareMathOperator{\sh}{sh}
\DeclareMathOperator{\ch}{ch}

\newcommand{\e}{\epsilon}

\title{The heat kernel of the asymmetric quantum Rabi model}
\author{Cid Reyes-Bustos}

\begin{document}

\maketitle

\begin{abstract}
  In this paper we derive an explicit formula for the heat kernel of the asymmetric quantum Rabi model (AQRM),
  a symmetry breaking generalization of the quantum Rabi model (QRM). The method described here is an extension of the recently
  developed one  for the heat kernel of the QRM based on the Trotter-Kato formula. In particular,
  the method is not based on path integrals or stochastic methods.  In addition to the heat kernel formula, we present 
  applications including the explicit formula for the partition function and the Weyl law for the distribution of the eigenvalues
  obtained from the analytic continuation of the corresponding spectral zeta function.
\end{abstract}

%\tableofcontents

\section{Introduction}
\label{sec:introduction}
  
The quantum Rabi model (QRM) is widely recognized as one of the fundamental models in quantum optics and the study of its properties has been in the spotlight in theoretical and experimental physics for a number of years. One of the main drivers is the possible application of the QRM to the development quantum computing and quantum information sciences. In parallel, there is a growing interest in the mathematical study of the properties of the QRM and its spectrum (see e.g. \cite{BdMZ2019,KRW2017,Sugi2016, KW2019}).

Notably, in \cite{RW2020hk} the explicit (or analytical) formula for the heat kernel of the QRM was obtained based on the Trotter-Kato product formula. The derivation involves the use of several techniques not common to this type of computation, such as the the use of harmonic analysis on the family of finite groups $\{\Z_2^n\}_{n\geq 0}$ or an extensive combinatorial (graph theoretical) discussion. The resulting expression is an power series with coefficients given by certain multiple integrals, which are interpreted as orbits of the action of the infinite symmetric group on the inductive limit of the groups $\Z_2^n, n \geq 0$, or as a as a type of discrete path integral \cite{RW2020z}. We note that conventional approaches using, for instance methods based on Feynman integrals or Feynman-Kac formulas, have not produced fully explicit formulas for the heat kernel of the QRM.

A careful examination of the computation of the heat kernel for the QRM reveals that it may be generalized to systems (physical or mathematical) other than the QRM Hamiltonian. Therefore, the understanding of the scope and limitations of the method is a significant topic of research in both in mathematical and theoretical physics. In this paper, as an starting point for this program, we extend the computation of the heat kernel to one of the simplest, yet significant, generalizations of the QRM, the asymmetric quantum Rabi model (AQRM).

The Hamiltonian of the AQRM is given by
\[ 
  \HRabi{\e} := \omega a^{\dagger}a + \Delta \sigma_z + g (a + a^{\dagger}) \sigma_x + \e \sigma_x,
\]
where, as usual,
\[
  \sigma_x =
  \begin{bmatrix}
    0 & 1  \\ 1 & 0
  \end{bmatrix}, \qquad
  \sigma_z =
  \begin{bmatrix}
    1 & 0 \\ 0 & -1
  \end{bmatrix},
\]
are the Pauli matrices, $a^{\dag}$ and $a$ are the creation and annihilation operators of the quantum harmonic oscillator with frequency $\omega$ (in this paper we set $\omega=1$) , i.e. $[a,a^{\dag}]=1$, $\e \in \R$ and $\Delta,g>0$. The QRM Hamiltonian is recovered by taking $\e  = 0$.

The AQRM was introduced as a ``symmetry breaking'' generalization of the QRM in \cite{B2011PRL}, where the proof of the exact solvability of both models was presented. We remark that it is the  $\Z_2$-symmetry in the QRM Hamiltonian that allows the presence of degeneracies in the spectrum in a natural way. In fact, the presence of the bias parameter breaks the natural symmetry of the QRM and makes the spectrum of the AQRM multiplicity free. The asymmetric model is considered to provide a more realistic description of the circuit QED experiments employing flux qubits in comparison with the QRM \cite{ni2010,YS2018}. 

Recall that heat kernel $K^{(\e)}_{\text{R}}(x,y,t)$ of the AQRM is the integral kernel corresponding to the operator \(e^{-t \HRabi{\e}} \) (one-parameter semigroup), that is, $K^{(\e)}_{\text{R}}(x,y,t)$ satisfies
\[
e^{-t \HRabi{\e}}\phi(x)= \int_{-\infty}^\infty  K^{(\e)}_{\text{R}}(x,y,t) \phi(y)dy
\]
for a compactly supported smooth function $\phi: \R \to \C^2$.  Equivalently, $K^{(\e)}_{\text{R}}(x,y,t)$ is the matrix-valued function satisfying
\[
  \frac{\partial}{\partial t} K^{(\e)}_{\text{R}}(x,y,t)= -\HRabi{\e} K^{(\e)}_{\text{R}}(x,y,t)
\]
for $t>0$ and $\lim_{t \to 0}K_{\text{R}}(x,y,t)=\delta_x(y) \bf{I}_2$ for $x,y \in \R$.

The formula developed in this paper (Theorem \ref{sec:main}) shows that the heat kernel is given as uniformly convergent series on the parameter $\Delta$ similar to the of the QRM, and that the parts corresponding to the bias parameter $\e$ in each coefficient of the series appear independent of the other parameters.

The main contribution of this paper is develop the tools for the generalization of the heat kernel computation for the AQRM. In particular, new difficulties arise since the non-commutative part in the computation is no longer a diagonal matrix. The scalar part of the computation remain largely the same as the case of the QRM, so in order to keep the exposition short and to avoid repeating the computation of the heat kernel, we focus on the new features and we refer the reader to \cite{RW2020hk} for the details.   The new features in the computation actually appear in the computation of the heat kernel for more general models and thus the result of this paper is a significant step towards the full understanding and generalization of the method.

Two immediate applications of the heat kernel formula is the formula for the time propagator, which is obtained by analytic continuation of the heat kernel to the imaginary line, and the explicit formula for the partition function of the AQRM. The propagator formula is expected to provide more precise computations for time evolution in AQRM (see \cite{WKB2012}), but explicit numerical studies are yet to be performed.

In this paper, we also consider further applications of the explicit formulas for the heat kernel and partition function of the AQRM obtained by means of the spectral zeta function.  The spectral zeta function of a physical model allows the study of the spectrum from the viewpoint of number theory with applications both to number theory and physics \cite{EE2012}. The (Hurwitz-type) spectral zeta function \(\zeta^{(\e)}_{\text{R}}(s; \tau)\) is given by the Dirichlet series
\[
  \zeta^{(\e)}_{\text{R}}(s; \tau):= \sum_{j=1}^\infty (\lambda^{(\e)}_j +\tau)^{-s},
\]
for $\Re(\tau)$ large enough. Here, $\lambda^{(\e)}_i$ are the (ordered) eigenvalues in the spectrum of $\HRabi{\e}$. In this paper, the meromorphic continuation is obtained by identifying the spectral zeta function of the AQRM as the Mellin transform of the partition function of the AQRM and then changing the path of integration in an appropriate way.

As an application we obtain the Weyl law for the distribution of the eigenvalues of the AQRM showing that asymptotically the  distribution does not depend on the bias parameter. We also complete the proof that the $G$-function of the AQRM is essentially given by the spectral determinant obtained from the spectral zeta function, first shown in \cite{KRW2017} under the assumption of the analytic continuation of the spectral zeta function of the AQRM by the method of \cite{Sugi2016}, which was not yet proved at the time.

Let us describe the structure of this paper. First, in Section \ref{sec:main} we give the main results of the paper, that is, the explicit formulas for the heat kernel and partition function of the AQRM. In Section \ref{sec:analyt-cont-spectr} we show the analytic continuation of the spectral zeta function of the AQRM and some consequences including the Weyl law.
The rest of the paper is devoted to the proof of the heat kernel formula. In Section \ref{sec:proof} we give a summary of the general method of computation. In Section \ref{sec:Hamilt} we give some general remarks on the asymmetric quantum Rabi model and make the initial computations with the Trotter-Kato product formula to obtain a limit formula that resembles a Riemann sum. Next, in Section \ref{sec:four-analys-limit} using Fourier analysis in the family of finite groups $\Z_2^N$ we transform certain sums in the limit into multiple Riemann integrals
allowing us to obtain a second limit expression that can be evaluated as a Riemann integral. The evaluation of the Riemann sum completes the proof of the explicit formula for the heat kernel and then the partition function is obtained directly as a corollary.

In light of the hidden symmetry of the AQRM, it would be interesting to consider the formulas in this paper when the parameter is half-integer. Indeed, recall that in the case of the AQRM there is no obvious way to define subspaces whose respective spectral graphs intersect creating degeneracies in the spectrum of the AQRM. It was shown in \cite{KRW2017} (based on previous works \cite{LB2015JPA,W2016JPA}) that degeneracies may appear for certain choices of parameters $g$ and $\Delta$ when the parameter $\e$ is half-integer. In \cite{MBB2020}, the symmetry operator for small values of half-integer $\e$ was obtained explicitly along with a general method of computation for the arbitrary half-integer $\e$ case. A more systematic approach to the half-integer case is given in \cite{RBW2021}.

\section{Main results}
\label{sec:main}

In this section we give the main results of this paper, the explicit formulas for the heat kernel and the partition
function for the AQRM. The proof of the formulas is the main contribution of this paper and are given in detail in
Section \ref{sec:proof}.

In this paper, we denote
by
\[
  \lambda^{(\e)}_1 < \lambda^{(\e)}_2 \leq \lambda^{(\e)}_3 \leq \ldots \leq \lambda^{(\e)}_n \leq \ldots (\nearrow \infty)
\]
the eigenvalues of \(\HRabi{\e}\). For \(\lambda = 0 \), we use the notation
\[
  \idotsint\limits_{0\leq \mu_1 \leq \cdots \leq \mu_\lambda \leq 1} f(x) d \bm{\mu_0} = f(x),
\]
for any function \(f\).

\begin{thm} \label{thm:heat}
    The heat kernel $K^{(\e)}_{\text{R}}(x,y,t) $ of the AQRM is given by the uniformly convergent series
  \begin{align*}    
    &K^{(\e)}_{\text{R}}  (x,y,t) =  K_0(x,y,g,t) \Bigg[ \sum_{\lambda=0}^{\infty} (t\Delta)^{\lambda} e^{-2g^2 (\coth(\tfrac{t}2))^{(-1)^\lambda}}
    \\
    &\quad \times \idotsint\limits_{0\leq \mu_1 \leq \cdots \leq \mu_\lambda \leq 1}  \exp\left(4g^2 \frac{\cosh(t(1-\mu_\lambda))}{\sinh(t)}(\frac{1+(-1)^\lambda}{2}) + \xi_{\lambda}(\bm{\mu_{\lambda}},t)\right) \\
         &\qquad  \begin{bmatrix}
            (-1)^{\lambda} \cosh  &  (-1)^{\lambda+1} \sinh  \\
            -\sinh &  \cosh
          \end{bmatrix}
                  \left[ \theta_{\lambda}(x,y,\bm{\mu_{\lambda}},t) + \e \left( \eta_\lambda(\bm{\mu_{\lambda}},t) + t\right) \right] d \bm{\mu_{\lambda}} \Bigg],
  \end{align*}
  with \(\bm{\mu_0} := 0\) and \(\bm{\mu_{\lambda}}= (\mu_1,\mu_2,\cdots,\mu_\lambda)\) and \(d \bm{\mu_{\lambda}} = d \mu_1 d \mu_2 \cdots d \mu_{\lambda} \) for \(\lambda \geq 1\). Here,
  \begin{align*}
    K_0(x,y,g,t)
    & = \frac{e^{g^2t}}{\sqrt{\pi (1-e^{-2t})}} \exp\left( - \frac{1+e^{-2t}}{2(1-e^{-2t})} (x^2 + y^2) +  \frac{2 e^{-t} x y}{1-e^{-2t}} \right)\\
  \end{align*}
  and the functions \(\theta_{\lambda}(x,y, \bm{\mu_{\lambda}},t)\),$\xi_\lambda(\bm{\mu_{\lambda}},t)$ and $\eta_\lambda(\bm{\mu}_{\lambda},t)$ are given by
\begin{align*} %\label{eq:auxfunc}
  \theta_{\lambda}(x,y, \bm{\mu_{\lambda}},t) &:= \frac{2\sqrt{2} g e^{-t}}{1-e^{-2t}}\left( x (e^{t}+e^{- t}) - 2 y \right) \left( \frac{1-(-1)^{\lambda}}{2} \right) - \sqrt{2} g (x-y) \frac{1+e^{-t}}{1-e^{-t}} \\
                           & \quad +   \frac{2\sqrt{2} g e^{-  t}}{1-e^{-2 t}} (-1)^{\lambda} \sum_{\gamma=0}^{\lambda} (-1)^{\gamma} \Big[ x  (e^{t(1 -   \mu_{\gamma}) } + e^{ t( \mu_{\gamma} - 1)})  -  y  (e^{- t \mu_{\gamma} }+ e^{ t \mu_{\gamma}})  \Big] \nonumber \\
  \xi_\lambda(\bm{\mu_{\lambda}},t) &:=  -\frac{2g^2 e^{-t}}{1-e^{-2t}} \left(e^{\frac12t(1-\mu_\lambda)}-e^{\frac12 t(\mu_{\lambda}-1)}\right)^2 (-1)^{\lambda}  \sum_{\gamma=0}^{\lambda} (-1)^{\gamma} (e^{- t \mu_{\gamma} }+ e^{ t \mu_{\gamma}})  \\
                     &\qquad  - \frac{2 g^2 e^{-t} }{1-e^{-2 t}} \sum_{\substack{0\leq\alpha<\beta\leq \lambda-1\\ \beta - \alpha \equiv 1 \pmod{2}  }}  \left( (e^{t(1-\mu_{\beta+1})} +  e^{t(\mu_{\beta+1}-1)} )-(e^{t(1-\mu_{\beta})} +  e^{t(\mu_{\beta}-1)}) \right) \nonumber  \\
                           &\qquad \qquad \qquad \qquad \qquad \qquad \times ( (e^{t  \mu_{\alpha}} + e^{-t \mu_{\alpha}}) - (e^{t \mu_{\alpha+1}} + e^{-t \mu_{\alpha+1}})), \nonumber \\
   \eta_\lambda(\bm{\mu}_{\lambda},t) &:= -2  t (-1)^{\lambda} \sum_{\gamma=1}^{\lambda} (-1)^{\gamma} \mu_{\gamma}. \nonumber
\end{align*}
where we use the convention \( \mu_0 = 0 \) whenever it appears in the formulas above.
\end{thm}

It is worth noting the expression of the heat kernel the system parameters $g,\Delta,\e$ do not appear mixed in the
coefficients. Note also that since the action of $\e$ in the original Hamiltonian is just the displacement $\e \sigma_z$, in the heat kernel
formula the contribution of the integration variables $\mu_i$ with coefficients including $\e$ appear linearly inside the exponential.

Directly from the analytical formula it is possible to verify that the heat kernel is well-behaved with respect to
the spatial variables. The proof of the following proposition may be adapted directly from that of the QRM, given
in \cite{RW2020z}. 

\begin{prop}
  Let
  \[
    K^{(\e)}_{\text{R}}(x,y,t)  =
    \begin{bmatrix}
      k_{1,1}(x,y,t; g, \Delta,\e) & k_{1,2}(x,y,t; g, \Delta,\e)\\
      k_{2,1}(x,y,t; g, \Delta,\e) & k_{2,2}(x,y,t; g, \Delta,\e)
    \end{bmatrix}.
  \]
  Then, for fixed $g,\Delta, t>0$ and $\e \in \R$, there are positive constants $a,b$ such that
  \[
    |k_{i,j}(x,y,t; g, \Delta,\e)| \leq a e^{-b (x^2+y^2)},
  \]
  for $i,j \in \{1,2\}$.
\end{prop}

The proposition shows that $K^{(\e)}_{\text{R}}(x,y,t)$ is a continuous function with respect to the spacial variables $x,y$. Similar results may be obtained with respect to the time variable $t$ and the system parameters.

By considering the analytic continuation of the heat kernel given in Theorem \ref{thm:heat} and the change of variable $t \to i t$, we obtain a formula for the integral kernel of the time-evolution operator (i.e. the time or wave propagator) of the AQRM. We refer the reader to \cite{RW2020z} for the technical details for the case of the QRM which can be easily extrapolated to cover the AQRM (see also Proposition \ref{prop:holomorphy} below).

Next, we consider the explicit formula of the partition function of the AQRM, defined by
\[
  Z_{\text{R}}^{(\e)}(\beta):= \sum_{n=1}^{\infty} e^{-\beta \lambda^{(\e)}_n} = \text{Tr}[e^{-t \HRabi{\e}}].
\]
The explicit formula for the partition function then is obtained directly from that of the heat kernel in an elementary way. The proof of the explicit formula is given in Section \ref{sec:four-analys-limit}.

\begin{cor} \label{cor:partition}
  The partition function \( Z^{(\e)}_{\text{R}}(\beta)\) of the AQRM is given by
  \begin{align*}
    Z^{(\e)}_{\text{R}}(\beta) = &\frac{2 e^{g^2\beta}}{1-e^{- \beta}} \Bigg[ \ch(\e \beta) \\
     & \quad+ \sum_{\lambda=1}^{\infty} (\beta \Delta)^{2\lambda}\idotsint\limits_{0\leq \mu_1 \leq \cdots \leq \mu_{2 \lambda} \leq 1} \Theta_{2 \lambda}(g,\beta,\bm{\mu_{2 \lambda}}) \ch\left[\e \beta \left( 1-2 \sum_{\gamma=1}^{2\lambda} (-1)^{\gamma} \mu_{\gamma}\right) \right] \bm{\mu_{2 \lambda}}\Bigg],
  \end{align*}
  where
  \[
    \Theta_{2 \lambda}(g,\beta,\bm{\mu_{2 \lambda}}) =  \exp\left(-2g^2 \coth(\tfrac{\beta}2)+ 4g^2\frac{\ch(\beta(1-\mu_{2\lambda}))}{\sh(\beta)} +  \xi_{2 \lambda}(\bm{\mu_{2\lambda}},\beta) +\psi_{2\lambda}^{-}(\bm{\mu_{2 \lambda}},\beta)\right)
  \]
  with
\begin{align*} %\label{eq:auxfunc}
  \psi_\lambda^{-}(\bm{\mu_{\lambda}},t) &:=  \frac{4 g^2 }{\sh(t)}\left[ \sum_{\gamma=0}^{\lambda} (-1)^{\gamma} \sh\left(t\left(\tfrac12 - \mu_{\gamma}\right)\right)\right]^2\\
\end{align*}
for \(\lambda \geq 1\) and \(\bm{\mu_{\lambda}} = (\mu_1,\mu_2,\cdots,\mu_\lambda) \) and where \( \mu_0 = 0 \).
\end{cor}

In particular, from the explicit formula it is immediate to verify that
\[
  Z^{(-\e)}_{\text{R}}(\beta) = Z^{(\e)}_{\text{R}}(\beta).
\]
This property reflects the well-known fact that the spectrum of $\HRabi{\e}$ is equal to the spectrum of
$\HRabi{-\e}$ (see e.g. Proposition 5.2 of \cite{KRW2017}).

In the next section we give applications of the explicit formulas of the partition function by means of the
spectral zeta function associated to the AQRM.

\begin{rem}
  When $g =0$, we verify that the spectrum of the AQRM is given by $n \pm \mu$, where $\mu = \sqrt{\Delta^2+\e^2}$. Also, in the proof of Theorem \ref{thm:heat} the parameter $\mu$ appears in a distinguished position, so it may be argued that a formulas for the heat kernel and partition functions given by power series in the parameter $\mu$ (in place of $\Delta$) would be more natural. This requires a nontrivial modification to the method and proof and is therefore not considered in this paper.
\end{rem}

\section{Applications to the spectrum of the AQRM via spectral zeta functions}
\label{sec:analyt-cont-spectr}

In this section, we prove the analytic continuation for the spectral zeta function of the AQRM to obtain applications to the study
of the spectrum of the AQRM. We note that the additional terms appearing in the partition function formula do not change the convergence
properties in a significant way. For a more detailed exposition, we refer the reader to \cite{RW2020z} (Appendix A).

The (Hurwitz-type) spectral zeta function \(\zeta^{(\e)}_{\text{R}}(s; \tau)\) is defined by the Dirichlet series
\[
  \zeta^{(\e)}_{\text{R}}(s; \tau):= \sum_{j=1}^\infty (\lambda^{(\e)}_j +\tau)^{-s}.
\]
We verify (see e.g. \cite{Sugi2016} for the QRM case) that the spectral zeta function $\zeta^{(\e)}_{\text{R}}(s; \tau)$
is absolutely convergent for \(\Re(s)>1 \) for \( \tau \in \C - \Spec({\HRabi{\e}})\) with $\Re(\tau) > \Delta + g^2 + |\e|$.

In the region of absolute convergence, we have a Mellin transform representation of the spectral zeta function
given by
 \begin{align*} %\label{SZF}
   \zeta_{\text{R}}^{(\e)}(s;\tau) & = \frac1{\Gamma(s)}\int_0^\infty t^{s-1}Z^{(\e)}_{\text{R}}(t)e^{-t\tau}dt.
\end{align*}

In order to obtain useful applications of the spectral zeta function, we need to extend its domain of definition to
the largest possible. It is the relation between the partition function and the spectral zeta function that allows us to give
a proof for the analytic continuation.

Let us define the function \( \Omega^{(\e)}(t)\) by
\begin{align*}
  \Omega^{(\e)}(t)& := (1-e^{-t}) Z_{\text{R}}^{(\e)}(t) 
\end{align*}

The main result of this section is the path integral expression for the spectral zeta function that gives the analytic continuation
to the complex plane with a simple pole at $s=1$. 

\begin{thm} \label{IntRep_SZF}
  For $\tau > \Delta + g^2 + |\e|$, we have
  \begin{align}\label{ContourSZF}
    \zeta_{\text{R}}^{(\e)}(s;\tau)= -\frac{\Gamma(1-s)}{2\pi i}\int_\infty^{(0+)} \frac{(-w)^{s-1} \Omega^{(\e)}(w)e^{-\tau w}}{1-e^{-w}}dw.
  \end{align}
  Here the contour integral is given by the path which starts at $\infty$ on the real axis, encircles the origin (with a radius smaller than $2\pi$) in the positive direction and returns to the starting point and it is assumed $|\arg(-w)|\leq \pi$. This gives the meromorphic continuation of $\zeta_{\text{R}}(s;\tau)$  to the whole plane where the only singularity is a simple pole with residue $2$ at $s=1$. \qed
\end{thm}

The theorem follows from the following estimate for $\Omega^{(\e)}(t)$, which we list here for completeness. The additional terms depending on $\e$ are easily bounded and the theorem follows as in the case of QRM (see \cite{RW2020z}).

\begin{prop} \label{prop:holomorphy}
  The series defining the function $\Omega^{(\e)}(t)$ is uniformly convergent in compacts in the
  complex domain $\mathcal{D}$ consisting a union of a half plane $\Re t>0$ and a disc centered at origin with
  radius $r < \pi $. In particular, $\Omega^{(\e)}(t)$ is a holomorphic function in the region \(\mathcal{D} \).
\end{prop}

A consequence of the meromorphic continuation of $\zeta_{\text{R}}^{(\e)}(s;\tau)$ is the Weyl law for the distribution of the
eigenvalues of the parity Hamiltonians \(\HRabi{\e}\) obtained by the use of Tauberian theorems in the usual
way (cf. \cite{IW2005a,IW2005b,Sugi2016}).

The distribution of the eigenvalues is described by the spectral counting function, given by
\begin{align*}
  N_{\text{R}}^{(\e)}(T) &= \# \{\lambda \in \Spec(\HRabi{\e}) \, | \, \lambda \le T \}, \\
\end{align*}
for \(T > 0 \).
The value of the residue of the pole of the spectral zeta function given in Theorem \ref{IntRep_SZF} determines the
asymptotic of the spectral counting function.

\begin{cor}
  We have
  \[
    N_{\text{R}}^{(\e)}(T)\sim 2T,
  \]
  as \(T \to \infty \). \qed
\end{cor}

In particular, the distribution of eigenvalues does not depend on the symmetry breaking parameter $\e$. Note that this result provides evidence for the conjecture of Braak for the distribution of energy levels of the QRM and the corresponding extension for the AQRM (see \cite{B2011PRL}).

Another application of the meromorphic continuation is the result given in \cite{KRW2017} on the spectral determinant of the AQRM. The analytic continuation at the point $s=0$ allows us to consider the notion of zeta-regularized product and spectral determinant, a generalization of the characteristic polynomial for operators.

The spectral determinant of AQRM is defined as
\begin{equation}\label{eq:ZRProduct}
  \det (\tau-\HRabi{\e}):= \regprod_{i=0}^\infty (\tau-\lambda_i):= \exp\big(-\frac{d}{ds}\zeta_{\HRabi{(\e)}}(s,\tau)\big|_{s=0}\big).
\end{equation}

The spectral determinant is an entire function that vanishes exactly at $\tau = \lambda \in \Spec(\HRabi{\e})$.
In \cite{KRW2017} it was shown that spectral determinant obtained from the analytic continuation of the spectral zeta
function is essentially equivalent (up to a non-vanishing factor) to the $G$-function obtained in the studies of
exact-solvability of the AQRM \cite{LB2015JPA}. However, the result was conditional on the analytic continuation of the
spectral zeta function of the AQRM to $s=0$, which was not proved at the time.

\begin{cor}[\cite{KRW2017}]
  There exists an entire non-vanishing function $c_{\e}(\tau ; g,\Delta)$ such that
  \begin{equation}
    \det (\tau-g^2-\HRabi{\e}) =  c_{\e}(\tau;g,\Delta) \mathcal{G}_{\e}(\tau;g,\Delta).
  \end{equation}
\end{cor}

Here, $\mathcal{G}_{\e}(\tau;g,\Delta)$ is the generalized $G$-function of the AQRM defined as
\begin{equation}\label{eq:neqG}
  \mathcal{G}_\e(x;g,\Delta) := G_\e(x;g,\Delta)\Gamma(\e-x)^{-1}\Gamma(-\e-x)^{-1},
\end{equation}
where $G_\e(x;g,\Delta)$ is the usual $G$-function of the AQRM. Different from the usual $G$-function $\mathcal{G}_\e(x;g,\Delta)$ is
an entire function which zeros correspond to the eigenvalues of AQRM, not only the regular spectrum. We note that
the function $\mathcal{G}_\e(x;g,\Delta)$ was originally considered for numerical computations in \cite{LB2015JPA} in a
truncated form.

A deeper study of the relation between the exact solvability of interaction models and the analytic continuation
of the corresponding spectral zeta functions is outside of the scope of the present paper.

\section{Proof of the heat kernel formula}
\label{sec:proof}

Before starting the proof of the main result, we briefly describe the method of the computation of the heat kernel,
dividing the process into a number of steps.

In the first step, the Hamiltonian $\HRabi{\e}$ is written as
\[
  \HRabi{\e} = H_1 + H_2,
\]
in such a way that each $H_i$ for $i=1,2$ satisfies the hypothesis of the Trotter-Kato formula and such that the
heat kernel can be explicitly computed. In the case of the AQRM, by means of a Bogoliubov transformation the choice
of operators is natural. In particular, $H_1$ is a type of non-commutative quantum harmonic oscillator.

The second step is to find an expression for the integral kernel of the operator
\[
  \left(e^{-t H_1} e^{-t H_2} \right)^N
\]
for arbitrary $N \in \Z_{\geq 1}$. To do this we have to consider both the scalar and matrix-valued part of the heat kernel by separate. In the case of the AQRM, the scalar value part is identical to that of the QRM (and is evaluated by Gaussian integration) but the matrix-value part requires a different method for the computations. Once the $N$-th power kernel is computed in an elementary way, the sum defining the heat kernel is rearranged by using the commutation structure of the matrices involved in the Hamiltonian, in a way that resembles a Riemann-sum (see \eqref{eq:limitfirst}). However, it cannot be evaluated directly as a Riemann integral due to oscillation of the signs in the summands.

To overcome this problem, we use the Fourier transform in finite groups $\Z_2^k$, for $k\geq 1$ (see e.g. \cite{Cecc2008}). In particular, using Parseval formula it is possible to control the oscillation of the sign in the dual stage with respect to certain $\Z_2$-vector invariant. This step may be interpreted as a transformation into a type of radial functions (see the comments following Def. 3.6 in \cite{RW2020hk}). For the AQRM, the Fourier analysis is considerably more complicated than the case of the QRM since intermediate expressions cannot be computed directly.

Once the oscillation of the signs is controlled, it remains to transform certain expression (given in terms of groups  $\Z_2^k$) into iterated integrals using a standard Riemann-Stieltjes method and transform the limit expression into a Riemann-type sum (see \eqref{eq:secondlimit}) that can be evaluated in a straightforward way completing the computation of the heat kernel.

\subsection{Hamiltonian factorization and initial considerations} %\label{sec:hamiltonian}
\label{sec:Hamilt}

The main tool behind the computation of the heat kernel of the AQRM is the Trotter-Kato formula
(see e.g \cite{Calin2011,Kato1978}, also known as Trotter formula). First, we write the Hamiltonian
$\HRabi{\e}$ as the sum of two simpler Hamiltonians whose heat kernels can be easily computed. It is then
immediate that
\[
  \HRabi{\e} = b^{\dagger} b - g^2 + \Delta \sigma_z + \e \sigma_x,
\]
with \( b = b(g) := a +  g \sigma_x \). We note that the operator $b^{\dag}b$ may be regarded as a non-commutative (or displaced) version of the quantum harmonic oscillator since the operators $b,b^\dag$  satisfy $[b,b^{\dag}]=\bm{I}_2$.
Thus, the spectrum of $b^{\dagger} b - g^2$ is given by
\[
  \Spec\left( b^{\dagger} b - g^2 \right) = \left\{ n - g^2 \, \big| \, n \in \Z_{\geq 0} \right\},
\]
where each eigenvalue has multiplicity $2$.

The operators \(b^{\dagger}b - g^2\) and \(\Delta \sigma_z + \e \sigma_x \) satisfy the conditions of the Trotter-Kato product formula and
we have
\[
  e^{- t \HRabi{\e}} = e^{- t (b^{\dagger}b -g^2 + \Delta \sigma_z + \e \sigma_x)} = \lim_{N\to \infty} (e^{-t (b^{\dagger}b -g^2)/N} e^{-t(\Delta \sigma_z + \e \sigma_x )/N})^N,
\]
in the strong operator topology.

The next step is to compute the integral kernel of the operator $e^{-t (b^{\dagger}b -g^2)} e^{-t(\Delta \sigma_z + \e \sigma_x )}$. This is done
in the standard way by using the Schwartz kernel and the Mehler's formula for the quantum Harmonic oscillator (see e.g.  \cite{Calin2011}).

\begin{prop} \label{prop:local_kernel}
The integral kernel \(D(x,y,t)\) for \(e^{-t (b^{\dagger}b -g^2)} e^{-t(\Delta \sigma_z + \e \sigma_x) }  \) is given by
\begin{align*}
 D(x,y,t) 
 =  \frac{u^{g^2}}{\sqrt{\pi (1-u^2)}} & \exp\left( -\frac{1-u}{1+u} \frac{((x+y)^2 + 8 g^2)}{4} -
      \frac{1+u}{1-u}\frac{(x-y)^2}{4} \right)  \\
      & \times \exp\left(- \frac{1-u}{1+u} \sqrt{2} g (x+y) \sigma_x \right) u^{\Delta \sigma_z + \e \sigma_x},
\end{align*}
with \(u=e^{-t}\) \qed.
\end{prop}

Next, we need  to compute the integral kernel \(D_N(x,y,t) \) of the operator
\[
  \left(e^{-t (b^{\dagger}b -g^2)} e^{-t(M(\e,\Delta))} \right)^N
\]
defined by the integral
\begin{equation}
  \label{eq:multiIntegral}
 D_N(x,y,t) =  \int_{-\infty}^{\infty} \cdots \int_{-\infty}^{\infty} D(x,v_1,t) D(v_1,v_2,t) \cdots D(v_{N-1},y,t) d v_{N-1} d v_{N-2} \cdots d v_1.  
\end{equation} 

By using the elementary identity
\begin{align*}
  \cosh\left(\alpha\right) \mathbf{I} - \sinh\left(\alpha\right)\mathbf{J} = \frac12\left( \mathbf{I} + \mathbf{J}\right) e^{-\alpha}
  + \frac12\left( \mathbf{I} - \mathbf{J}\right) e^{\alpha},
\end{align*}
valid for arbitrary $\alpha \in \C$, to expand the matrix terms in the expression of $D(x,y,t)$ we rewrite
\(D_N(x,y,t) \) as a sum
\begin{align*} %\label{eq:Npkernel}
  D_N(x,y,t) = \sum_{\bm{s} \in \Z_2^{N}} G^{(\e)}_N(u,\Delta,\bs) I_N(x,y,u, \, \bs),
\end{align*}
for a scalar function $ I_N(x,y,u, \, \bs)$ and a $2\times 2$ matrix-valued function $G^{(\e)}_N(u,\Delta,\bs)$ for which we give
explicit expressions below. The group $\Z_2$ appearing as a set in the sum above takes an important role in the computation of the heat kernel.

Since the function $ I_N(x,y,u, \, \bs)$ does not depend on the bias parameter $\e$, it has the same expression as in the symmetric case. In particular, we see that all the integrals in \eqref{eq:multiIntegral} are contained in $ I_N(x,y,u, \, \bs)$ and these can be evaluated by multivariate Gaussian integration.

Let us recall the notation used in the expression of \(I_N(x,y,u, \,\bs) \). For \(\bs \in \Z_2^{N}\) and \(i,j \in \{1,2,\cdots,N\}\), define
\begin{gather*} %\label{eq:notation}
  \eta_i(\bs) := (-1)^{\bs(i)}+(-1)^{\bs(i+1)},\\
  \Lambda^{(j)}(u) := u^{j-1} \left( 1 - u^{2 (N-j) +1} \right), \qquad   \Omega^{(i,j)}(u) =  u^{j-i} \left(1-u^{2 i} \right) \left( 1 -u^{2 (N - j)} \right). \nonumber
\end{gather*}

\begin{thm}[Theorem 2.5 of \cite{RW2020hk}] \label{thm:IN}
  For \(N \in \Z_{\geq 1} \), we have
  \begin{align*} \label{eq:IN}
  &I_N(x,y,u, \,\bs) =  K_0(x,y, g, u^N) \exp \left( \sqrt{2} g \frac{(1-u)}{(1-u^{2 N})}  \sum_{j=1}^{N} (-1)^{\bs(j)}\left(  x \Lambda^{(j)}(u) + y \Lambda^{(N-j+1)}(u)  \right) \right)   \nonumber   \\
          &\times  \exp \left(    \frac{  g ^2 (1-u)^2}{2(1+u)^2 (1-u^{2 N})} \bigg(  \sum_{i=1}^{N-1} \eta_i(\bs)^2 \Omega^{(i,i)}(u)  + 2 \sum_{i<j} \eta_i(\bs) \eta_j(\bs) \Omega^{(i,j)}(u) \bigg) -   \frac{2 N g^2 (1-u)}{1+u} \right). \qed
\end{align*}
\end{thm}

To simplify later computations, we set
\[
  \bar{I}_N(x,y,u, \,\bs) := I_N(x,y,u, \,\bs)/ K_0(x,y, g, u^N).
\]

In general, throughout the computation of the heat kernel, we see that the computations related to the scalar parts of the heat kernel
do not largely differ to the symmetric case.

\subsubsection{Non-commutative part}

In contrast with the scalar part, since the matrix $\Delta \sigma_z + \e \sigma_x$ is not diagonal, the analysis of the non-commutative part is more involved than in the QRM case. In this section, we first obtain an explicit expression of $G^{(\e)}_N(u,\bs)$ and give the expression of the heat kernel as the limit of a Riemann-type sum.

Let 
\[
   \bm{I} :=
   \begin{bmatrix}
     1 & 0 \\
     0 & 1
   \end{bmatrix}, \qquad
   \bm{J} :=
   \begin{bmatrix}
     0 & 1 \\
     1 & 0
   \end{bmatrix},
 \]
then the function $G^{(\e)}_N(u,\bs)$ is given by 
\[
  G^{(\e)}_N(u,\bs) :=  \frac{1}{2^N} \overrightarrow{\prod}_{i=1}^{N}[\bm{I}+(-1)^{1-\bs(j)}\bm{J}] u^{\Delta \sigma_z + \e \sigma_x}.
\]

To simplify the computations, it is convenient to consider the diagonalization of $\Delta \sigma_z + \e \sigma_x$. Set
\[
  M = M(\e,\Delta) := \Delta \sigma_z + \e \sigma_x,
\]
then we verify that
\begin{equation}
  \label{eq:eigen}
  C M C^{-1} =   \begin{bmatrix}
    -\mu & 0 \\ 0 & \mu
 \end{bmatrix},
\end{equation}
with \(\mu = \sqrt{\e^2 + \Delta^2} \) and
\[
  C = C(\e,\Delta) :=
  \begin{bmatrix}
    \Delta-\mu & \e \\ \Delta+\mu & \e.
  \end{bmatrix}.
\]
The parameter $\mu$ appears in several places of the computation, in a similar way to the parameter $\Delta$ in the
corresponding computation of the QRM heat kernel.

\begin{rem}
  It is important to note that for the case $\e=0$, the expression above may not be valid. The problem is that
  in the entries of the matrix $C^{-1}$, the parameter $\e$ appears in the denominator. However, this is not a problem
  since in this case the matrix $\Delta \sigma_z$ is already a diagonal matrix (so in this case we take $C=\mathbf{I}$).
  Moreover, the formulas in this section agree with the case $\e =0$ with a limit interpretation.
\end{rem}

Next, we define two functions, one for the scalar and one for the  matrix part and  of $G^{(\e)}_N(u,\bs)$.
For $v,w \in \{0,1\}$, the function  $h_{v,w}(\tau)$ is given by
\[
  h_{v,w}(\tau) :=  \left(\frac{1+(-1)^{v+w}}2 \right) (1+\tau) +  \left(\frac{(-1)^v+(-1)^{w}}2 \right) \frac{\e}{\mu} (1-\tau) + \left(\frac{1-(-1)^{v+w}}{2} \right) \frac{\Delta}{\mu} (1-\tau).
\]

\begin{ex} \label{ex:funH}
  The function $h_{v,w}(\tau)$ has a simple form for the specific values of $v,w \in \{0,1\}$. Namely,
  \begin{align*}
    h_{0,0}(\tau) &= \frac1{\mu}\left( \mu (1+\tau) + \e (1-\tau) \right), \qquad h_{1,1}(\tau) = \frac1{\mu}\left( \mu (1+\tau) - \e (1-\tau) \right)\\
    h_{0,1}(\tau) &= h_{1,0}(\tau) = \frac{\Delta}{\mu} (1-\tau),
  \end{align*}
  we note that the function \(h_{v,w}(\tau)\) is invariant under permutation of the variables \(v\) and \(w\).
\end{ex}

On the other hand, for \(\bs \in \Z_2^{k}\), we define the matrix-valued function
\[
  \mathbf{M}_k(\bs)  := \mathbf{M}_{i j},
\]
where \(\bs(1) = i \) and \(\bs(k) = j \) for \(i,j = 0,1 \), and
\begin{equation*}
  \label{eq:fourmatrices}
  \mathbf{M}_{0 0} := \matrixZZ, \mathbf{M}_{0 1} := \matrixZO, \mathbf{M}_{1 0} := \matrixOZ, \mathbf{M}_{1 1} := \matrixOO.
\end{equation*}

Next, we obtain the explicit form of  $G^{(\e)}_k(u,\bs)$.

\begin{prop}
  For $\bs \in \Z_2^{k}$, we have
  \begin{equation}
    \label{eq:Gcomp}
        G^{(\e)}_k(u,\bs) =  \frac{ \prod^{k-1}_{i=1}h_{s(i),s(i+1)}(u^{2\mu})}{u^{(k-1)\mu} 2^k} \mathbf{M}_k(\bs) u^{\Delta \sigma_z + \e \sigma_x},
  \end{equation}
\end{prop}

\begin{proof}
  For $v,w,\rho \in \{0,1\}$, it is enough to verify
  \[
    \mathbf{M}_{v,w} C^{-1} u^{-\mu \sigma_z} C (\bm{I} + (-1)^{1-\rho} \bm{J}) = \frac{h_{w,\rho}(u^{2\mu})}{u^{\mu}} \mathbf{M}_{v,\rho},
  \]
  and the main result then follows by induction.
\end{proof}

As in the case of QRM, we give a special definition for the scalar part of $G^{(\e)}_k(u,\bs)$ for later use.
\begin{dfn} 
  For \(k \geq 1 \), the function \(g_k(u,\bs) \) is given by
  \begin{equation*}
    g_k(u,\bs) = \frac{ \prod^{k-1}_{i=1}h_{s(i),s(i+1)}(u^{2\mu})}{u^{(k-1)\mu} 2^k}.
  \end{equation*}
\end{dfn}

Let us summarize the computation up to this point and write the limit expression corresponding to the heat kernel.
By definition, the heat kernel is given by
\[
  K_{\text{R}}^{(\e)}(x,y,t) = \lim_{N\to \infty} \sum_{\bm{s} \in \Z_2^{N}} G^{(\e)}_N(u^{1/N},\bs) I_N(x,y,u^{1/N}, \bs).
\]

In order to evaluate this limit, we transform it into a Riemann-type sum that we can later evaluate.
A key observation is that the matrices in the non-commutative part depend only on the first and last entry of the
vectors of $\Z_2^N$, thus it is convenient to define subsets where the first and last entries of its elements are fixed.
  
\begin{dfn} %\label{dfn:subsets}
  Let \(N \in \Z_{\geq 1} \) and \(i,j \in \Z_2\). The subset
  \(\mathcal{C}^{(N)}_{i,j} \subset \Z_2^{N} \) is given by
  \[
    \mathcal{C}^{(N)}_{i,j} = \{ \bs \in \Z_2^{N} \, |\,  \bs(1) = i, \bs(N) = j \}.
  \]
\end{dfn}

In practice, we consider partitions of $\Z_2^N$ with fixed tails of ones or zeros and restrict the function
\(\bar{I}_N(x,y, u^{1/{N}},\bs)\) for such a subset. For instance, for \(k \geq 1\), we write
\begin{align*} %\label{eq:defJ}
  \bar{I}_N(x,y,u, \,\bar{\bs} \oplus \bm{0}_{N-k+1} ) &= J^{(k,N)}_0(x,y,u,g) R_0^{(k,N)}(u,g,\bar{\bs}),  \nonumber\\
  \bar{I}_N(x,y,u, \,\bar{\bs} \oplus \bm{1}_{N-k+1} ) &= J^{(k,N)}_1(x,y,u,g) R_1^{(k,N)}(u,g,\bar{\bs}),
\end{align*}
with functions \(J^{(k,N)}_\mu(x,y,u,g)\) and \( R_\mu^{(k,N)}(u,g,\bar{\bs}) \) for \(\mu \in \{0,1\}\) defined implicitly (see Def. 3.3 of \cite{RW2020hk} for the full expression). Note that \( \bar{\bs} \in \mathcal{C}^{(k-1)}_{i,1} \) and in the second line \( \bar{\bs} \in \mathcal{C}^{(k-1)}_{i,0} \) for \(i=0,1\).

With these preparations (see also Definition 2.1 in \cite{RW2020hk} for details), we obtain the limit expression for $K_{\text{R}}^{(\e)}(x,y,t)$, namely, it is equal to the sum of
\begin{align*}
    &  \frac12 K_0(x,y,g,u) \lim_{N \to \infty} \Bigg(  \left(\frac{h_{0,0}(u^{\frac{2\mu}N})}{2 u^{\frac{\mu}N}} \right)^{N-1} J^{(1,N)}_0(x,y,u^{\frac1N},g) \bm{M}_{0,0}  \\
  &\qquad \qquad \qquad \qquad \qquad +\left(\frac{h_{1,1}(u^{\frac{2\mu}N})}{2 u^{\frac{\mu}N}} \right)^{N-1} J^{(1,N)}_1(x,y,u^{\frac1N},g) \bm{M}_{1,1} \Bigg)
\end{align*}
and
\begin{align}
  \label{eq:limitfirst}
  &K_0(x,y,g,u) \lim_{N \to \infty} \left( \frac{h_{0,1}(u^{\frac{2\mu}N})}{2 u^{\frac{\mu}N}} \right) \nonumber \\
  &\times \sum_{i=0}^i \sum_{k \geq 2}^N \Bigg[ \left( \frac{h_{i,i}(u^{\frac{2\mu}N})}{2 u^{\frac{\mu}N}} \right)^{N-k} J_i^{(k,N)}(x,y,u^{\frac{1}N},g)\sum_{v=0}^1 \bm{M}_{v i} \sum_{ \bs \in \setC{k-1}{v}{1}} g_{k-1}(u^{\frac1N},\bs) R_i^{(k,N)}(u^{\frac{1}N},\bs)\Bigg) u^{\frac{M(\e,\Delta)}N}.
\end{align}

We notice that since $h_{0,0}(\tau) \neq h_{1,1}(\tau)$ for $\e \neq 0$, the expression of the limit for the AQRM splits according to the tail of zeros or ones. Actually, this difference turns out to be only apparent and the limit expression simplifies
later.

Next, we use the harmonic analysis in the groups $\Z_2^N$ to transform the sum above into one that can be evaluated as a Riemann sum by controlling the oscillation of the signs in the expression of $I_N(x,y,u, \bs)$. The main tool is the application of the Parseval identity
on the abelian groups $\Z_2^{k-3}$ to the sum
\begin{equation}
  \label{eq:sumPar}
  \sum_{ s \in \setC{k-1}{v}{w}} g_{k-1}(u^{\frac1N}, \bs) R_w^{(k,N)}(u^{\frac1N},\bs),
\end{equation}
by taking advantage of the fact that the set $\setC{k-1}{v}{w}$ may be equipped with a $\Z_2^{k-3}$ abelian group
structure.

The use of the Fourier transform also allows to identify the limit above as an orbit sum over the action of the
infinite symmetric group $\mathfrak{S}_\infty$ on the (induced limit of) finite groups $\Z_2^N$ for $N \geq 0$. We refer the
reader to \cite{RW2020z} for the description of this interpretation.

\subsection{Fourier analysis and limit evaluation}
\label{sec:four-analys-limit}

In order to evaluate sums of the type \eqref{eq:sumPar}, we need to describe the Fourier transform of \(g_k(u,\bs) \)
and $R_w^{(k,N)}(u,\bs)$. This is a necessary step in order to use Parseval identity and ultimately rewrite the limit
expression \eqref{eq:limitfirst} into one that can be evaluated as a Riemann integral. A good reference for the Fourier
transform on finite groups (and particularly in the abelian group $\Z_2^k$) can be found in \cite{Cecc2008}.

\begin{dfn} %\label{dfn:varphi}
  Let \(\rho = (\rho_1,\rho_2,\cdots,\rho_k) \in \Z_2^{k} \). The function \(\vert \cdot \vert : \Z_2^{k} \to \C \) is
  given by
  \begin{equation*} %\label{eq:norm}
    \vert \rho \vert = \| \rho \|_1 := \sum_{i=1}^{k} \rho_i.
  \end{equation*}  
  Let \(j_1 < j_2 < \cdots < j_{\vert \rho \vert} \) the position of the ones in \(\rho\), that is, \(\rho_{j_i} = 1\) for
  all \(i \in \{ 1,2,\cdots,\vert \rho \vert \}\) and if \(\rho_i =1\) then \(i \in \{j_1,j_2,\cdots,j_{\vert \rho \vert}\} \).
  The function  \( \varphi_k : \Z_2^{k} \to \C \) is given by
  \begin{equation*} %\label{eq:phi}
    \varphi_k(\rho)  := \sum_{i=1}^{\vert \rho \vert} (-1)^{i-1} j_{\vert \rho \vert + 1 - i}
    = j_{\vert \rho \vert} - j_{\vert \rho \vert-1} +  \cdots + (-1)^{\vert \rho \vert-1} j_1,
  \end{equation*}
  and \( \varphi_k(\bZ) =0 \) where \(\Z_2 \) is the identity element in \(\Z_2^{k}\). For \(k=0\), define
  \(\varphi_k(\rho) = |\rho| = 0 \) where \(\rho\) is the unique element of \(\Z_2^{0}\).
\end{dfn}

We remark here that the function \(\varphi\) depends only on the numbers \(j_1,j_2,\cdots,j_\ell\) (the position of the ones) for a given vector \(\rho\)
with $|\rho|=\ell$.

In the case of the function \(g_k(u,\bs) \) we define a special notation. 

\begin{dfn} %\label{dfn:g_k} 
 Let \(v,w \in \{0,1\} \). Then, for \(\bs \in \Z_2^{k}\) with \(k \geq 1 \), define the function \(g^{(v,w)}_k(u,\bs)\) by
  \begin{equation*}
    %\label{eq:gkdef}
    g^{(v,w)}_k(u,\bs) := h_{v,s(1)}(u^{2\mu})h_{s(k),w}(u^{2\mu}) \prod^{k-1}_{i=1}h_{s(i),s(i+1)}(u^{2\mu}).
  \end{equation*}
  In addition, for \(\rho \in \Z_2^{0}\), define
  \[
    g_0^{(v,w)}(u,\bs) = h_{v,w}(u^{2\mu}).
  \]
\end{dfn}

For \(\bs \in \setC{k+2}{v}{w}\), we have
\begin{equation*} %\label{eq:equiv_gk_gkvw}
  2^{k+2} u^{(k+1)\mu} g_{k+2}(u,\bs) =  g^{(v,w)}_k(u,\bar{\bs}),
\end{equation*}
where for $s \in \Z^{k+2}_2$, $\bar{s} \in \Z_2^k$ is the projection of $s$ obtained by removing the first and
last component. Note that the degree of \(g^{(v,w)}_k(u,\bar{\bs})\) as a polynomial in \(u^{\mu}\) is  \(2(k+1)\).

Next, we obtain the Fourier transform of \(g^{(v,w)}_k(u,\bar{\bs})\). 

\begin{prop} %\label{prop:fourier_g}
  For \( \rho \in \Z_2^{k}\), we have
  \begin{equation*}
    \left[
      \widehat{g^{(v,0)}_k}(\rho), \,\,
      \widehat{g^{(v,1)}_k}(\rho) \right]
    = \left[ h_{0,v}(u^{2\mu}),\,\, h_{1,v}(u^{2\mu}) \right] \overrightarrow{\prod}_{i=1}^{k} \bm{B}(\rho_i),
  \end{equation*}
  where the matrix-valued function \(\bm{B}(s)\), for \(s \in \Z_2 \), is given by by
\[
  \bm{B}(s) =
  \begin{bmatrix}
    h_{0,0}(u^{2\mu}) & h_{0,1}(u^{2\mu}) \\
   (-1)^s h_{1,0}(u^{2\mu}) & (-1)^s h_{1,1}(u^{2\mu})
  \end{bmatrix}.
\]
\end{prop}

\begin{proof}
  The case \(k=0 \) is trivial. For \(k \geq 1 \), let \( \rho = (\rho_1,\rho_2,\cdots,\rho_k) \in \Z_2^{k} \)
  and \( \delta = (\rho_1,\rho_2,\cdots,\rho_{k-1}) \in \Z_2^{k-1}\), then we have
  \begin{align*}
    &\left[
      \widehat{g^{(v,0)}_k}(\rho), \,\, \widehat{g^{(v,1)}_k}(\rho) \right] = \left[\sum_{\bm{s}\in \Z_2^{k+1}} g_{k+1}^{(v,0)}(\bs) \chi_{\rho} (\bm{s}) , \sum_{\bm{s} \in \Z_2^{k+1}} g_{k+1}^{(v,1)}(\bs) \chi_{\rho} (\bm{s})  \right] \\
    &\qquad = \left[ h_{0,0}(u^{2\mu}) \widehat{g^{(v,0)}_{k-1}}(\delta) + (-1)^{\rho_k} h_{0,1}(u^{2\mu}) \widehat{g^{(v,1)}_{k-1}}(\delta)  , h_{0,1}(u^{2\mu}) \widehat{g^{(v,0)}_{k-1}}(\delta) + (-1)^{\rho_k} h_{1,1}(u^{2\mu}) \widehat{g^{(v,1)}_{k-1}}(\delta)  \right] \\
    &\qquad = [\widehat{g^{(v,0)}_{k-1}}(\delta), \,\, \widehat{g^{(v,1)}_{k-1}}(\delta) ] \bm{B}(\rho_{k+1}),
  \end{align*}
  and the result follows by induction.
\end{proof}

\begin{rem}
  Notice that in the formula for the Fourier transform of the function $\widehat{g^{(v,0)}_k}(\rho)$ the function $\varphi(s)$
  does not appear directly as in the case of the QRM.
\end{rem}

Rewriting, we obtain
\[
  \widehat{g^{(v,w)}_k}(\rho) = [h_{0,v}(u^{2\mu}), h_{1,v}(u^{2\mu}) ] \, \overrightarrow{\prod}_{i=1}^{k-1} \bm{B}(\rho_i)
\begin{bmatrix}
  h_{0,w}(u^{2\mu}) \\
  (-1)^{\rho_k} h_{1,w}(u^{2\mu})
\end{bmatrix}.
\]

Next, we consider the sum
\[
  \sum_{\rho \in \Z_2^{k}} \widehat{g^{(v,w)}_k}(\rho) \prod_{i=1}^k(A_i)^{\rho_i},
\]
for \(A_i \in \C\). In this case, while the result resembles the case of the QRM, the method of proof is different. In particular, the matrices in the expression of $\widehat{g^{(v,w)}_k}(\rho)$ correspond in some sense to the recurrence relation appearing in the QRM case.
In the next expression, for a vector \(\rho = (\rho_1,\rho_2,\cdots,\rho_k) \in \Z_2^k\) with \(|\rho|=\ell \), the numbers \( j_1 < j_2 <\cdots < j_\ell\) correspond to the positions of the ones in \(\rho\).  %We refer the reader to Lemma 3.4 of \cite{RW2020hk} for the precise statement and discussion.

\begin{prop} \label{prop:fouriersum}
  We have
  \begin{align*}
    \sum_{\rho \in \Z_2^{k}} \widehat{g^{(v,w)}_k}(\rho) \prod_{i=1}^k(A_i)^{\rho_i} = &\sum_{\ell=0}^{k} h_{v, \ell+ w}(u^{2\mu})(h_{0,1}(u^{2\mu}))^\ell 
                                                                       (h_{1-w,1-w}(u^{2\mu}))^{k-\ell}  \\
    &\quad \times\sum_{\substack{ \rho \in \Z_{2}^k \\ |\rho|=\ell}} \left(\frac{h_{w,w}(u^{2\mu})}{h_{1-w,1-w}(u^{2\mu})}\right)^{\alpha(\rho)}
  \prod_{i=0}^{\ell} \left( \prod_{n= j_i +1}^{j_{i+1}} (1+(-1)^{w+\ell+i} A_n) \right),
  \end{align*}
  where \(j_0:= 0 \) and \(j_{\ell+1}:=k \). The function \(\alpha(s) \) is given by
  \[
     \alpha(s) = k - \left\lfloor \frac{|s|}2 \right\rfloor - \varphi(s).
  \]
\end{prop}

%Clearly, we recover Prop. 3.12 (up to the ordering of the matrix factors) of the QRM case by setting \(\e= 0\), making \(c= 0\).

\begin{proof}
  First, let us define the function
  \[
    F_k(\tau) = \sum_{\rho \in \Z_2^k}  \overrightarrow{\prod}_{i=1}^{k} A_i^{\rho_i} \bm{B}(\rho_i).
  \]
  Clearly, we have
  \[
    F_k(\tau) = F_{k-1}(\tau) \left( \bm{B}(0) +  A_i\bm{B}(1) \right),
  \]
  and, by the elementary identity,
  \[
    \left( \bm{B}(0) +  A_i\bm{B}(1) \right) = \bm{D}(A_i) \bm{B}(0),
  \]
  with
  \[
    \bm{D}(x) :=
    \begin{bmatrix}
      1 + x & 0 \\
      0 & 1 - x
    \end{bmatrix},
  \]
  it follows that
  \[
    F_k(\tau) = \overrightarrow{\prod}_{i=1}^{k} \bm{D}(A_i) \bm{B}(0).
  \]
  Next, we write
  \[
    \bm{B}(0) = \frac{1}{\mu}(a \bm{I} + b \bm{J} + c \bm{K}),
  \]
  with
  \[
    a =  \mu(1+u^{2\mu}), \quad b= \Delta (1- u^{2\mu}), \quad c = \e (1-u^{2\mu}),
  \]
  and
  \[
    \bm{K} := \sigma_z = 
    \begin{bmatrix}
      1 & 0 \\
      0 & -1 
    \end{bmatrix},
  \]
  and then
  \[
    F_k(\tau) = \frac{1}{\mu^{k}} \overrightarrow{\prod}_{i=1}^{k} \bm{D}(A_i) \left( a \bm{I} + b \bm{J} + c \bm{K} \right),
  \]
  and using the commutation relations
  \[
    \bm{I} \bm{D}(x) = \bm{D}(x) \bm{I}, \qquad \bm{J} \bm{D}(x) = \bm{D}(-x) \bm{J}, \qquad \bm{K} \bm{D}(x) = \bm{D}(x) \bm{K},
  \]
  we obtain
  \begin{align} \label{eq:fk}
    F_k(\tau) = \frac{1}{\mu^{k}} \sum_{s \in \Z_2^k} (b\bm{J})^{|\rho|} (a \bm{I} - b \bm{K})^{k-|\rho|-\alpha(s)}(a \bm{I} + b \bm{K})^{\alpha(s)} \bm{D}
    \left(\prod_{i=0}^{\ell} \left( \prod_{n= j_i +1}^{j_{i+1}} ((-1)^{\ell+i} A_n)\right)\right),
  \end{align}
  with $\alpha(s)$ given by $\alpha((0)) := 1$, $\alpha((1)):= 0$ and recursively as
  \begin{align*}
    \alpha(s\oplus(0)) &= \alpha(s) +1, \\
    \alpha(s\oplus(1)) &= k - |s| - \alpha(s)
  \end{align*}
  for $s \in \Z_2^{k+1}$. We verify directly that $\alpha(s) = k - \left\lfloor \frac{|s|}2 \right\rfloor - \varphi(s) $ by comparing the recurrence relations
  on both sides.
  Next, we write
  \begin{align*}
    \sum_{\rho \in \Z_2^{k}} &\widehat{g^{(v,w)}_k}(\rho) \prod_{i=1}^k(A_i)^{\rho_i}  = [h_{0,v}(u^{2\mu}), h_{1,v}(u^{2\mu}) ] \sum_{\rho \in \Z_2^{k}} \, \overrightarrow{\prod}_{i=1}^{k} \bm{B}(\rho_i) A_i^{\rho_i}
                       \begin{bmatrix}
                         \delta_0(w) \\
                         \delta_1(w)
                       \end{bmatrix},
  \end{align*}
  and the result is obtained by substituting equation \eqref{eq:fk} and multiplying out the matrices.
\end{proof}

It remains to compute the Fourier transform of the function $R^{(v,w)}_\mu(\bs)$. Here, since there parameter $\e$ does not appear there is no change from the computation in the QRM case. Let us just introduce some of the basic notations
and refer to the appropriate results when needed. We write the function $R^{(v,w)}_{\mu}(\bs)$
\[
  R^{(v,w)}_{\mu}(\bs) = \exp\left( a_0^{(\mu)} + \sum_{\chi \in S_{k-3}} a_\chi^{(\mu)} \chi(\bs) \right),
\]
where the coefficients $a_\chi^{(\mu)}$ are implicitly define (See Lemma 3.5 of \cite{RW2020hk} for the explicit form).
Then, the Fourier transform of $R^{(v,w)}_{\mu}(\bs)$ is given by
\begin{equation*} %\label{eq:FourierR1}
  \widehat{R^{(v,w)}_{\mu}}(\bs) =  2^{k-3 }\exp(a_0^{(\mu)} ) \left(\sum_{\xi \in \widetilde{\Z_2^{k-3}}} C_\xi^{(\mu)} \delta_{\xi,\chi_{\rho}}\right).
\end{equation*}
where
\[
  C_{\chi_\rho } = \sum_{\substack{\br \in \{0,1\}^{\ell} \\ \chi_\rho = \prod_{i=1}^{\ell} (\chi_i)^{\br_i}}} T^{(\br)}(\mathbf{a}).
\]
Here, for \(\mathbf{a} \in \C^{\ell} \) and an index vector \(\br \in \{0,1\}^\ell \) we define
\begin{align*} %\label{eq:Tfunc}
  T^{(\br)}(\mathbf{a}) &:= \prod_{i=1}^{\ell} \left[ \cosh(\mathbf{a}_i)^{(1-\br_i)}\sinh(\mathbf{a}_i)^{\br_i} \right], 
\end{align*}
where \(\mathbf{a}_i \) (resp. \(\br_i\)) denotes the \(i\)-th component of \( \mathbf{a}\) (resp. \(\br \)).
We remark that the particular form of the $T^{(\br)}(\mathbf{a})$ function is a consequence of the fact that all characters
of the group $\Z_2^{k}$ are real characters.

We now return to the evaluation of the heat kernel. Let $\eta \in \{0,1\}$, then, by Parseval identity, we have
\begin{align*}
  \sum_{\bs \in \setC{k-1}{v}{w}} g_{k-1}(\bs)  R^{(k,N)}_\eta(\bs)&= \frac{1}{2^{k-1} u^{(k-2)\mu}} \sum_{\bs \in \Z_2^{k-3}} g^{(v,w)}_{k-3}(\bs) R^{(v,w)}_\mu(\bs) \\
              &= \frac{1}{2^{2(k-1)} u^{(k-2)\Delta }} \sum_{\rho \in \Z_2^{k-3}} \widehat{g^{(v,w)}_{k-3}}(\rho) \widehat{R_\mu^{(v,w)}}(\rho) \\
         &= \frac{1}{2^{(k-1)} u^{(k-2)\mu }} \exp(a_0^{(\eta)}) \sum_{\rho \in \Z_2^{k-3}} \widehat{g^{(v,w)}_{k-3}}(\rho) \sum_{\br \in V^{(k-3)}_{\bZ}}  T^{(\sigma_{\rho}(\br))}(\mathbf{a}^{(\eta)}),
\end{align*}
notice that the extra $2^{(k-1)}$ in the determinant corresponds to the factors $h_{\alpha,\beta}(u^{2\mu})$ ($\alpha,\beta \in \{0,1\}$) inside the
matrices $B(\rho)$ of $\widehat{g^{(v,w)}_{k-3}}(\rho)$. Next, we have
\begin{align*}
  \sum_{\rho \in \Z_2^{k-3}} \widehat{g^{(v,w)}_{k-3}}(\rho) \sum_{\bm{r} \in V^{(k-3)}_{\bm{0}}} T^{(\sigma_{\rho}(\br))}(\mathbf{a}^{(\eta)})  = \sum_{\bm{r} \in V^{(k-3)}_{\bm{0}}} T^{(\bm{r})}(\mathbf{a}^{(\eta)}) \sum_{\rho \in \Z_2^{k-3}}  \widehat{g^{(v,w)}_{k-3}}(\rho) \prod_{i=1}^{k-3} \left( \tanh(a_i^{(\eta)})^{1-r_{0 i}}\right)^{\rho_i},
\end{align*}
and we apply Proposition \ref{prop:fouriersum}, to get
\begin{align*}
  & \sum_{\br \in V^{(k-3)}_{\bZ}} T^{(\br)}(\mathbf{a}^{(\eta)}) \sum_{\ell=0}^{k-3} h_{v, \ell+ w}(u^{2\mu})(h_{0,1}(u^{\frac{2\mu}N}))^\ell 
    (h_{1-w,1-w}(u^{2\mu}))^{k-3-\ell}  \\
  & \qquad \qquad \qquad \qquad \qquad \times\sum_{\substack{ \rho \in \Z_{2}^{k-3} \\ |\rho|=\ell}} \left(\frac{h_{w,w}(u^{2\mu})}{h_{1-w,1-w}(u^{2\mu}}\right)^{\alpha(\rho)}   
  \prod_{i=0}^{\ell} \left( \prod_{n= j_i +1}^{j_{i+1}} (1+(-1)^{w+\ell+i} A_n^{(r)}) \right) \\
  =& \sum_{\ell=0}^{k-3} h_{v, \ell+ w}(u^{2\mu})(h_{0,1}(u^{\frac{2\mu}N}))^\ell  (h_{1-w,1-w}(u^{2\mu}))^{k-3-\ell}  \\
  & \qquad \qquad \qquad  \times \sum_{\substack{ \rho \in \Z_{2}^{k-3} \\ |\rho|=\ell}} \left(\frac{h_{w,w}(u^{2\mu})}{h_{1-w,1-w}(u^{2\mu}}\right)^{\alpha(\rho)}   
  \sum_{\br \in V^{(k-3)}_{\bm{Z}}} T^{(\br)}(\mathbf{a}^{(\eta)})  \prod_{i=0}^{\ell} \left( \prod_{n= j_i +1}^{j_{i+1}} (1+(-1)^{v_0 + v_i} A_n^{(r)}) \right),
\end{align*}
and by using Lemma 3.14 of \cite{RW2020hk}, we obtain
\begin{align*}
  &\sum_{\substack{ \rho \in \Z_{2}^{k-3}}} h_{v, |\rho|+ w}(u^{2\mu}) \left( h_{1-w,1-w}(u^{2\mu})\right)^{k-3} \left(\frac{h_{0,1}(u^{\frac{2\mu}N})}{h_{w,w}(u^{2\mu})} \right)^{|\rho|}
  \\
  & \qquad \left(\frac{h_{w,w}(u^{2\mu})}{h_{1-w,1-w}(u^{2\mu})}\right)^{\alpha(\rho)}\exp\left(\sum_{m=0}^{k-4}  \sum_{j=1}^{k-3-m} (-1)^{(|\rho|+w)\delta_0(m) + \sum_{i=m}^{m+j-1} \rho_i} a_{m, m+j}^{(\eta)} \right).
\end{align*}
 Summing, up, we see that 
 \begin{align}
   \label{eq:dualsum}
  \sum_{\bs \in \setC{k-1}{v}{w}} &g_{k-1}(\bs)  R^{(k,N)}_\eta(\bs) \nonumber \\
  =& \frac{1}{2^{(k-1)} u^{(k-2)\mu }} \left(  h_{1-w,1-w}(u^{2\mu})\right)^{k-3} \sum_{\substack{ \rho \in \Z_{2}^{k-3}}} h_{v, |\rho|+ w}(u^{2\mu})
     \left(\frac{h_{0,1}(u^{\frac{2\mu}N})}{ h_{1-w,1-w}(u^{2\mu})}\right)^{|\rho|} 
   \nonumber \\
  & \qquad \left(\frac{ h_{w,w}(u^{2\mu})}{ h_{1-w,1-w}(u^{2\mu})}\right)^{\alpha(\rho)}\exp\left(a_0^{(\eta)} + \sum_{m=0}^{k-4}  \sum_{j=1}^{k-3-m} (-1)^{(|\rho|+w)\delta_0(m) + \sum_{i=m}^{m+j-1} \rho_i} a_{m, m+j}^{(\eta)} \right).
\end{align}

\subsubsection{Second limit expression and final computation} %\label{sec:second-limit-expr}

In the next step, we reformulate the limit expression \eqref{eq:limitfirst} of the heat kernel using the Fourier
analysis developed in the previous section. As a result, in the dual space the sign changes are controlled when we fix the length of the vectors in the the Fourier transformed expression, thus allowing us to evaluate the limit as a Riemann integral.

First, by an elementary computation we see that
\begin{equation}
  \label{eq:loglim}
    \log\left(\frac{h_{i,i}(e^{-t\frac{2\mu}N})}{2 e^{-t\frac{\mu}N}} \right) = (-1)^i \frac{t \e}{N} + O\left(\tfrac{1}{N^2} \right),
\end{equation}
for $i=0,1$ and thus we have the limits
\[
  \lim_{N\to\infty} \left(\frac{h_{0,0}(u^{\frac{2\mu}N})}{2 u^{\frac{\mu}N}} \right)^{N-1} = u^{-\e}, \qquad \lim_{N\to\infty} \left(\frac{h_{1,1}(u^{\frac{2\mu}N})}{2 u^{\frac{\mu}N}} \right)^{N-1} = u^{\e}.
\]
Using these limits, we rewrite the limit expression for the heat kernel as
\begin{align}
  \label{eq:secondlimit}
  & K_0(x,y,g,u) \Bigg\{ e^{-2g^2 \frac{1-e^{-t}}{1+e^{-t}}}
                     \begin{bmatrix} \cosh & -\sinh
                       \\ -\sinh & \cosh
                     \end{bmatrix}
                                   \Big( \sqrt2 g(x+y)\frac{1-e^{-t}}{1+e^{-t}} + t \e \Big) \nonumber  \\
  & + \frac{u^{-\e}}{2} \lim_{N \to \infty} \left( \frac{h_{0,1}(u^{\frac{2\mu}N})}{2 u^{\frac{\mu}N}} \right) \sum_{k \geq 2}^N J_0^{(k,N)}(x,y,u^{\frac{1}N},g)
     \sum_{\substack{ \rho \in \Z_{2}^{k-3}}} \frac{1}{u^{\frac{\mu}N}}
    \begin{bmatrix}
      \alpha_{|\rho|+1}(u^{\frac{2\mu}N}) & -\alpha_{|\rho|+1}(u^{\frac{2\mu}N})\\
      -\beta_{|\rho|+1}(u^{\frac{2\mu}N}) & \beta_{|\rho|+1}(u^{\frac{2\mu}N})
    \end{bmatrix}
       \nonumber \\
  &\, \times  \left(\frac{h_{0,1}(u^{\frac{2\mu}N})}{h_{0,0}(u^{\frac{2\mu}N})}\right)^{|\rho|} \left(\frac{h_{1,1}(u^{\frac{2\mu}N})}{h_{0,0}(u^{\frac{2\mu}N})}\right)^{\alpha(\rho)} \exp{\left(a_0^{(0)} + \sum_{m=0}^{k-4}  \sum_{j=1}^{k-3-m} (-1)^{(|\rho|+w)\delta_0(m) + \sum_{i=m}^{m+j-1} \rho_i} a_{m, m+j}^{(0)} \right)} \nonumber \\
  & + \frac{u^{\e}}{2} \lim_{N \to \infty} \left( \frac{h_{0,1}(u^{\frac{2\mu}N})}{2 u^{\frac{\mu}N}} \right) \sum_{k \geq 2}^N
     J_1^{(k,N)}(x,y,u^{\frac{1}N},g)
     \sum_{\substack{ \rho \in \Z_{2}^{k-3}}} \frac{1}{u^{\frac{\mu}N}}
  \begin{bmatrix}
     -\alpha_{|\rho|}(u^{\frac{2\mu}N}) & -\alpha_{|\rho|}(u^{\frac{2\mu}N}) \\
      \beta_{|\rho|}(u^{\frac{2\mu}N}) & \beta_{|\rho|}(u^{\frac{2\mu}N})
    \end{bmatrix}
    \nonumber \\
 &\, \times  \left(\frac{h_{0,1}(u^{\frac{2\mu}N})}{h_{1,1}(u^{\frac{2\mu}N})}\right)^{|\rho|} \left(\frac{h_{0,0}(u^{\frac{2\mu}N})}{h_{1,1}(u^{\frac{2\mu}N})}\right)^{\alpha(\rho)}  \exp\left(a_0^{(1)} + \sum_{m=0}^{k-4}  \sum_{j=1}^{k-3-m} (-1)^{(|\rho|+w)\delta_0(m) + \sum_{i=m}^{m+j-1} \rho_i} a_{m, m+j}^{(1)} \right) \Bigg\}.
\end{align}
where
\[
  \alpha_x(\tau) = \frac{1}{2}(h_{0,x}(\tau) - h_{1,x}(\tau)), \qquad \qquad \beta_x(\tau) = \frac1{2}(h_{0,x}(\tau) + h_{1,x}(\tau)).
\]

Next, we simplify the expressions appearing in the limit above, starting with the matrices, to complete the computation of
the heat kernel . Note that since
\[
  \lim_{N\to \infty} h_{w,w}(u^{\frac{2\mu}N}) = 1, \qquad \lim_{N\to \infty} h_{w,1-w}(u^{\frac{2\mu}N}) = 0,
\]
we have
\[
  \lim_{N \to \infty} \frac{1}{u^{\frac{\mu}N}}
    \begin{bmatrix}
      \alpha_{|\rho|+1}(u^{\frac{2\mu}N}) & -\alpha_{|\rho|+1}(u^{\frac{2\mu}N})\\
      -\beta_{|\rho|+1}(u^{\frac{2\mu}N}) & \beta_{|\rho|+1}(u^{\frac{2\mu}N})
    \end{bmatrix}  =
    \begin{bmatrix}
      (-1)^{|\rho|+1} & (-1)^{|\rho|}\\
      -1 & 1
    \end{bmatrix},
\]
and
\[
  \lim_{N \to \infty}
  \frac{1}{u^{\frac{\mu}N}}
  \begin{bmatrix}
    -\alpha_{|\rho|}(u^{\frac{2\mu}N}) & -\alpha_{|\rho|}(u^{\frac{2\mu}N}) \\
    \beta_{|\rho|}(u^{\frac{2\mu}N}) & \beta_{|\rho|}(u^{\frac{2\mu}N})
  \end{bmatrix}  =
  \begin{bmatrix}
    (-1)^{|\rho|+1} & -(-1)^{|\rho|}\\
    1 & 1
  \end{bmatrix},
\]
and at the limit we see that the resulting matrices are the same that appear in the  QRM heat kernel.

Next, we consider the limit of the quotients of the functions $h_{w,v}(u)$. For $\ell \in \R$ fixed and $w \in \{0,1\}$,
we see directly that
\[
  \left(\frac{h_{w,w}(e^{-t\frac{2\mu}N})}{h_{1-w,1-w}(e^{-t\frac{2\mu}N})}\right)^{\ell} = e^{(-1)^w 2 t \ell \e / N} +  O\left(\frac1N\right),
\]
In particular, for $w \in \{0,1\}$ we have
\[
  \lim_{N\to \infty} \left(\frac{h_{w,w}(u^{\frac{2\mu}N})}{h_{1-w,1-w}(u^{\frac{2\mu}N})}\right)^{\ell} = 1.
\]
Similarly, for $\ell \in \R$ we observe that
\[
  \left(\frac{h_{0,1}(u^{\frac{2\mu}N})}{h_{1-w,1-w}(u^{\frac{2\mu}N})}\right)^{\ell} = \left(\frac{t \Delta}{N}\right)^\ell \left( e^{(-1)^w\frac{t \e \ell}{N}} + O\left(\frac1N\right) \right).
\]
These limits may be verified by considering the power series expansion of the logarithm as in \eqref{eq:loglim}.

Let us now consider the expressions appearing in the inner sums in the limit \eqref{eq:secondlimit}. For instance, for $w\in \{0,1\}$ and fixed $|\rho|=\ell$, from the expression
\[
  \left(\frac{h_{0,1}(u^{\frac{2\mu}N})}{h_{1-w,1-w}(u^{\frac{2\mu}N})}\right)^{|\rho|} \left(\frac{h_{w,w}(u^{\frac{2\mu}N})}{h_{1-w,1-w}(u^{\frac{2\mu}N})}\right)^{\alpha(\rho)},
\]
a direct computation gives
\begin{equation}
  \label{eq:simplifyeq}
  \left(\frac{t \Delta}{N}\right)^\ell  e^{\frac{(-1)^wt \e \ell}{N}} e^{2 (-1)^w t (k-3 - \lfloor\ell/2\rfloor - \varphi(\rho)) \e / N} +  O\left(\frac1N\right).
\end{equation}
In this case, since the function  $\varphi(\rho)$ appears in the part corresponding to the $\e$, we need to evaluate it using
multiple Riemann integral (see Section 3.4 of \cite{RW2020hk}).
Concretely, we see that 
\begin{align*}
  &\sum_{\substack{\bs \in \setC{k-1}{v}{w} \\ |\rho|=\ell}} \left(\frac{h_{1,1}(u^{\frac{2\mu}N})}{h_{0,0}(u^{\frac{2\mu}N})}\right)^{\alpha(\rho)} g_{k-1}(\bs)  R^{(k,N)}_\eta(\bs)  = \exp\left(-2(-1)^{\eta} \e \frac{t}{N}\right)\\
  &\quad \times \sum_{\substack{\bs \in \setC{k-1}{v}{w} \\ |\rho|=\ell}} \exp\left((-1)^\eta \frac{2t \varphi(\rho) \e}N\right)\exp\left(a_0^{(\eta)} + \sum_{m=0}^{k-4}  \sum_{j=1}^{k-3-m} (-1)^{(|\rho|+w)\delta_0(m) + \sum_{i=m}^{m+j-1} \rho_i} a_{m, m+j}^{(\eta)} \right).
\end{align*}
up to $O(\frac1N)$ terms. Then, the sum over $\setC{k-1}{v}{w} \simeq \Z_2^{k-3}$ with fixed $|\rho|=\ell$ is interpreted as a sum
over $j_1<j_2<\ldots < j_\ell \leq k-3$ by considering the position of the ones in the vectors $\rho$.

For \(\lambda \geq 1 \), define the functions
\begin{align*}
    f^{(\eta)}_\lambda(&z_1,z_2,\cdots,z_\lambda;u^{\frac1N}) =  (-1)^{\eta+1}  \frac{2\sqrt{2} g}{1-u^{2}} \sum_{\gamma=1}^{\lambda} (-1)^{\gamma-1} \Big[ x u^2 (1-u^{-2+\frac{z_{\lambda+1-\gamma}}N})(1-u^{-\frac{z_{\lambda+1-\gamma}}N}) \\
              &  \qquad \qquad \qquad \qquad \qquad \qquad \qquad \qquad \qquad \qquad \qquad - y u (1-u^{\frac{z_{\lambda+1-\gamma}}N})(1-u^{-\frac{z_{\lambda+1-\gamma}}N})  \Big] \\
              &  - \frac{2 g^2 }{1-u^2} u^{\frac{k}N}(1-u^{1-\frac{k}N})^2 \sum_{\gamma=1}^{\lambda} (-1)^{\gamma-1} (1-u^{\frac{z_{\lambda+1-\gamma}}N})(1-u^{-\frac{z_{\lambda+1-\gamma}}N}) \\
              & - \frac{2 g^2 }{1-u^2}\sum_{\substack{0\leq\alpha<\beta \\ \beta - \alpha \equiv 1 \pmod{2}  }}^{\lambda}  u^{\frac{z_{\beta+1}-z_{\alpha}}{N}}(1-u^{2 - \frac{z_{\beta+1}+z_{\beta}}N })(1-u^{\frac{z_\beta - z_{\beta+1}}N})(1-u^{\frac{z_\alpha - z_{\alpha+1}}{N}})(1-u^{\frac{z_\alpha+z_{\alpha+1}}N}),
\end{align*}
and
\[
  g^{(\eta)}_\lambda(z_1,z_2,\cdots,z_\lambda;\frac{t}{N}) = (-1)^\eta\frac{2 t \e}{N} \sum_{\gamma=1}^\lambda (-1)^{\gamma-1} z_{\lambda+1-\gamma}  
\]
where as before, we set \(z_0 := 0 \) and  \(z_{\lambda+1} := k-2 \). Notice that for fixed \( \lambda\),  \(f_\lambda^{(\eta)}(\bm{z};u^{\frac1N})\) and  $g^{(\eta)}_\lambda(z_1,z_2,\cdots,z_\lambda;\frac{t}{N})$ are smooth functions on \( z_i\), with \(i=1,2,\cdots,\lambda \), for any \( u \in (0,1) \).

To transform the sum into a multiple integral, we need the following result. The proof is done by using Riemann-Stieltjes integration as in the case of the QRM and we refer the reader to Section 3.4 of \cite{RW2020hk} for the details.
  
\begin{lem} \label{lem:sumint}
  For fixed \(\lambda \geq 1\) and \(a \in \Z_{\geq 1}\) with \(a \leq N \), we have
  \begin{align*}
    &  \sum_{1 \leq i_{1} < i_{2} < \cdots < i_{\lambda}}^{a} e^{f_\lambda^{(\eta)}(i_1,\cdots,i_\lambda;u^{\frac1N})+g^{(\eta)}_\lambda(i_1,\cdots,i_\lambda;\frac{t}{N}) } \\
    &\qquad \qquad \qquad=  \int_{0}^a \int_0^{z_\lambda} \cdots \int_0^{z_{2}} e^{{f_\lambda^{(\eta)}(z_1,\cdots,z_\lambda;u^{\frac1N})} +g^{(\eta)}_\lambda(z_1,\cdots,z_\lambda;\frac{t}{N})}  d \bm{z} + O(a^{\lambda-1}).
  \end{align*}
\end{lem}

With these preparations, the computation of the heat kernel of the AQRM can be completed by partitioning the sums
in \eqref{eq:secondlimit} according to the norm of $\rho$. The limit in a generic matrix component
of \eqref{eq:secondlimit} after the simplifications of this section is thus given by
\begin{align}
  &\frac{u^{\pm \e}}{2} \sum_{\lambda=0}^\infty \lim_{N \to \infty} \left( \frac{h_{0,1}(u^{\frac{2\mu}N})}{2 u^{\frac{\mu}N}} \right) \sum_{k \geq 2}^N J_0^{(k,N)}(x,y,u^{\frac{1}N},g) \exp\left(-2(-1)^{\eta} \e \frac{t}{N}\right)\\
  &\quad \times \sum_{\substack{\bs \in \Z_{2}^{k-3} \\ |\rho|=\lambda}} \exp\left((-1)^\eta \frac{2t \varphi(\rho) \e}N\right)\exp\left(a_0^{(\eta)} + \sum_{m=0}^{k-4}  \sum_{j=1}^{k-3-m} (-1)^{(\lambda+w)\delta_0(m) + \sum_{i=m}^{m+j-1} \rho_i} a_{m, m+j}^{(\eta)} \right),
\end{align}
for $\eta,v,w \in \{0,1\}$. Then, the innermost sum is replaced with a multiple integral using Lemma \ref{lem:sumint} and the final computation is obtained by evaluation of the Riemann sum in an straightforward way (see the proof of Theorem 4.2 in \cite{RW2020hk}). We note that the part of the sum corresponding to $g^{(\eta)}_\lambda(z_1,\cdots,z_\lambda;\frac{t}{N})$ is collected in $\eta_\lambda(\bm{\mu}_{\lambda},t)$ along with the additional term $2(-1)^{\eta}\e \frac{t}{N} $ appearing in \eqref{eq:simplifyeq}. This concludes the proof of Theorem \ref{thm:heat}.

Finally, we give the proof of Corollary \ref{cor:partition}, that is, the explicit expression for the partition
function of the AQRM. 

  \begin{proof}[Proof of Corollary \ref{cor:partition}]

    First, let us recall the identity
    \begin{equation*}
      Z^{(\e)}_{\text{R}}(\beta):=  \int_{-\infty}^\infty  \tr K^{(\e)}_{\text{R}}(x,x,\beta) dx.
    \end{equation*} -
    Directly by Theorem \ref{thm:heat}, we see that $\tr K^{(\e)}_{\text{R}}(x,x,t)$ is equal to
    \begin{align*}
   \frac{2 e^{g^2t} e^{-x^2 \frac{1-e^{-t}}{1+e^{-t}} }}{\sqrt{\pi (1-e^{-2t})}} 
  &\Bigg\{ e^{-2g^2\frac{1-e^{-t}}{1+e^{-t}}} \cosh\left(2\sqrt2 gx  \frac{1-e^{-t}}{1+e^{-t}}\right) \\
  &+  e^{-2g^2 \coth(\tfrac{t}2)} \sum_{\lambda =1}^{\infty} (t\Delta)^{2\lambda} \idotsint\limits_{0\leq \mu_1 \leq \cdots \leq \mu_{2 \lambda} \leq 1} e^{4g^2 \frac{\cosh(t(1-\mu_\lambda))}{\sinh(t)}+\xi_{2 \lambda}(\bm{\mu_{2\lambda}},t)}   \\
  &\times \cosh\left( \frac{2 \sqrt{2} g x}{1+e^{-t}} \sum_{\gamma=0}^{2\lambda} (-1)^{\gamma}  \left( e^{- t \mu_{\gamma} } - e^{ t( \mu_{\gamma}- 1)}  \right) + \e \left( \eta_\lambda(\bm{\mu_{\lambda}},t) + t\right) \right) d \bm{\mu_{2\lambda}}  \Bigg\}
    \end{align*}
    then the result follows from the elementary identity
    \[
      \int_{-\infty}^\infty e^{-\alpha x^2} \cosh(x \, \eta)dx = \sqrt{\frac{\pi}{\alpha}}e^{\frac{\eta^2}{4\alpha}}.
    \]
    valid for $\alpha>0$ and $\gamma,\eta \in \R$.
\end{proof}

\section*{Acknowledgements}

This work was supported by JST CREST Grant Numbers JPMJCR14D6 and JPMJCR2113, Japan.
The author would like to thank Nguyen Thi Hoai Linh for the useful comments on a early version of the manuscript.

\begin{flushleft}

\bigskip

 Cid Reyes-Bustos \par
  NTT Institute for Fundamental Mathematics,\\
  NTT Communication Science Laboratories, NTT Corporation \\
3-9-11, Midori-cho Musashino-shi, Tokyo, 180-8585, Japan \\
  \texttt{cid.reyes@ntt.com, math@cidrb.me}

\end{flushleft}


\begin{thebibliography}{99}

\bibitem{A2020}
  S.~Ashhab:
  \textit{Attempt to find the hidden symmetry in the asymmetric quantum Rabi model},
  Phys. Rev. A \textbf{101} (2020), 023808.

\bibitem{BdMZ2019}
  A.~Boutet de Monvel and L.~Zielinski:
  \textit{Oscillatory Behavior of Large Eigenvalues in Quantum Rabi Models}, Int. Math. Res. Notices, 
  DOI:10.1093/imrn/rny294, 2019 (first published online January 2019).
    
\bibitem{B2011PRL}
  D.~Braak:
  \textit{Integrability of the Rabi Model},
  Phys. Rev. Lett. \textbf{107} (2011), 100401.
  
\bibitem{B2013MfI}
  D.~Braak. 
  \textit{Analytical solutions of basic models in quantum optics},
  in ``Applications + Practical Conceptualization + Mathematics = fruitful Innovation, Proceedings of the Forum of Mathematics for Industry 2014" eds. R.~Anderssen, et al., 75-92, Mathematics for Industry \textbf{11}, Springer, 2016.
  
\bibitem{Calin2011}
  O.~Calin, D.-H.~Chang, K.~Furutani and C.~Iwasaki:
  Heat kernels for elliptic and sub-elliptic operators. Methods and techniques,
  Applied and Numerical Harmonic Analysis Series, Birkh\"auser, 2011.

\bibitem{Cecc2008}  
  T.~Ceccherini-Silberstein, F.~Scarabotti and F.~Tolli:
  Harmonic Analysis on Finite Groups: Representation Theory, Gelfand Pairs and Markov Chains,
  Cambridge Studies in Advanced Mathematics, Cambridge University Press, 2008. 

\bibitem{EE2012}
  E.~Elizalde:
  Ten Physical Applications of Spectral Zeta Functions,
  Lecture Notes in Physics 855, Springer Berlin, Heidelberg, 2012.
  
\bibitem{GD2013}
  B.~Gardas and J.~Dajka:
  \textit{New symmetry in the Rabi model},
  J. Phys. A: Math. Theor. \textbf{46} (2013), 265302.

\bibitem{IW2005a}
  T.~Ichinose and M.~Wakayama:
  \textit{Zeta functions for the spectrum of the non-commutative harmonic oscillators},
  Commun. Math. Phys. \textbf{258} (2005), 697--739.
  
\bibitem{IW2005b}
  T.~Ichinose and M.~Wakayama:
  \textit{Special values of the spectral zeta functions for the non-commutative harmonic oscillators and confluent Heun equations},
  Kyushu J. Math. \textbf{59} (2005), 39--100.
  
\bibitem{Kato1978}
  T.~Kato: 
  \textit{Trotter's product formula for an arbitrary pair of self-adjoint contraction semigroups}, Topics in functional analysis (essays dedicated to M.~G.~Kre\u\i n on the occasion of his 70th birthday), Adv. Math. Suppl. Stud., 3, Boston, MA: Academic Press, (1978) 185--195.

% \bibitem{Ivic1985}
%   A.~Ivi\'c:
%   The Riemann Zeta-Function: Theory and Applications, Dover, 1985.

\bibitem{KRW2017}
  K.~Kimoto, C.~Reyes-Bustos and M.~Wakayama:
  \textit{Determinant expressions of constraint polynomials and degeneracies of the asymmetric quantum Rabi model}.
  Int. Math. Res. Notices (2020), Published online 20 April 2020.

 \bibitem{KW2019}
  K.~Kimoto and M.~Wakayama:
 \textit{Ap\'ery-like numbers for non-commutative harmonic oscillators and automorphic integrals},
  Ann.  l'Inst. Henri Poincar\'e - D, 2022 (Published Online: December 2022).
  
\bibitem{LB2015JPA}
  Z.-M.~Li and M.T.~Batchelor:
  \textit{Algebraic equations for the exceptional eigenspectrum of the generalized Rabi model},
  J. Phys. A: Math. Theor. \textbf{48} (2015), 454005.g
  
\bibitem{MBB2020}
  V.~V.~Mangazeev, M.~T.~Batchelor and V.~V.~Bazhanov:
  \textit{The hidden symmetry of the asymmetric quantum Rabi model},
  J. Phys. A: Math. Theor. \textbf{54} (2021), 12LT01.
  
\bibitem{ni2010}
  T.~Niemczyk, {\it et al.}. 
  \textit{Beyond the Jaynes-Cummings model: circuit QED in the ultrastrong coupling regime},
  Nat. Phys. \textbf{6} (2010), 772-776.

\bibitem{RBW2021}
  C.~Reyes-Bustos, D.~Braak and M.~Wakayama:
  \textit{Remarks on the hidden symmetry of the asymmetric quantum Rabi model},
  J. Phys. A: Math. Theor. \textbf{54} (2021), 285202. 
  
\bibitem{RW2020hk}
  C.~Reyes-Bustos and M.~Wakayama:
  \textit{Heat kernel for the quantum Rabi model}, 
  Adv. Theor. Math. Phys. \textbf{26} (5), (2022), 1347-1447.

\bibitem{RW2020z}
  C.~Reyes-Bustos and M.~Wakayama:
  \textit{Heat kernel for the quantum Rabi model II: Propagators and spectral determinants},
  J. Phys. A: Math. Theor. \textbf{54} (2021), 115202.

\bibitem{RBW2022}
  C.~Reyes-Bustos and M.~Wakayama:
  \textit{Degeneracy and hidden symmetry for the asymmetric quantum Rabi model with integral bias},
  Comm. Numb. Theor. Phys. \textbf{16} (3),  (2022), 615-672.
    
\bibitem{Sugi2016}
  S.~Sugiyama:
  \textit{Spectral zeta functions for the quantum Rabi models},
  Nagoya Math.~J. \textbf{229} (2018), 52--98.
 % DOI: https://doi.org/10.1017/nmj.2016.62 (Published online: December 2016).
  
% \bibitem{T1951}
%   E. C.~Titchmarsh:
%   The Theory of the Riemann Zeta-function, Oxford Univ. Press 1951.

\bibitem{W2016JPA}
  M.~Wakayama. 
  \textit{Symmetry of asymmetric quantum Rabi models}. J. Phys. A: Math. Theor.
  \textbf{50} (2017), 174001 (22pp).
  
\bibitem{WKB2012}
  F.A.~Wolf, M.~Kollar, and D.~Braak:
  \textit{Exact real-time dynamics of the quantum Rabi model}.
  Phys. Rev. A \textbf{85} (2012), 053817.

\bibitem{XZBL2017}
  Q.-T.~Xie, H.-H.~Zhong, M.T.~Batchelor and C.-H.~Lee:
  \textit{The quantum Rabi model: solution and dynamics},
  J. Phys. A: Math. Theor. %- Semi-classical-and-quantum-Rabi-models
  \textbf{50} (2017), 113001.    
  
 \bibitem{YS2018}
F.~Yoshihara, T.~Fuse, Z.~Ao, S.~Ashhab, K.~Kakuyanagi, S.~Saito, T.~Aoki, K.~Koshino, and K.~Semba:  
\textit{Inversion of Qubit Energy Levels in Qubit-Oscillator Circuits in the Deep-Strong-Coupling Regime}, 
Phys. Rev. Lett. \textbf{120} (2018), 183601.
  
\end{thebibliography}
\end{document}